\documentclass[a4paper]{article}
\usepackage{amsfonts,amssymb,amsmath,amsthm,cite}
\usepackage{ucs} 
\usepackage[utf8x]{inputenc}
\usepackage{graphicx}
\usepackage[english]{babel}
\usepackage{slashed}
\usepackage{textcomp}
\usepackage{bm}
\usepackage{titlesec}
\usepackage{tikz-cd}
\usepackage{adjustbox}
\usepackage{multirow}
\usepackage{etoolbox}
\patchcmd{\thebibliography}{\section*}{\section}{}{}
\usepackage{mathtools}
\mathtoolsset{showonlyrefs}
\evensidemargin=\oddsidemargin
\newtheorem{theorem}{Theorem}
\newtheorem{lemma}{Lemma}
\newtheorem{cor}{Corollary}

\begin{document}
\vspace{15mm}
\begin{center}
	\large{\textbf{Quasi-local probability averaging \\in the context of cutoff regularization}}
\end{center}
\vspace{2mm}
\begin{center}
	\large{\textbf{A. V. Ivanov${}^{\dagger}$~~~I. V. Korenev${}^{\star}$}}
\end{center}
\begin{center}
	${}^{\dagger}$St. Petersburg Department
	of Steklov Mathematical Institute
	of Russian Academy of Sciences,\\
	27 Fontanka, St. Petersburg, 191023, Russia
\end{center}
\begin{center}
	${}^{\dagger}$Saint Petersburg State University,\\ 	7/9 Universitetskaya Emb., St. Petersburg, 199034, Russia
\end{center}
\begin{center}
	${}^{\star}$National Research University "Higher School of Economics", Faculty of Mathematics,\\6 Usacheva St., Moscow, 119048, Russia
\end{center}

\begin{center}
	${}^{\dagger}$E-mail: regul1@mail.ru\,\,\,
	${}^{\star}$E-mail: jacepool332@gmail.com
\end{center}

\vspace{15mm}

\textbf{Abstract.} In this paper, we study the properties of averaged fundamental solutions of a special type for Laplace operators in the Euclidean space of an arbitrary dimension. We consider a class of kernels suitable for probabilistic averaging, and propose new representations for the deformed fundamental solutions and their values at zero. In addition, we give examples related to specific quantum field models in the context of studying renormalization properties.

\vspace{2mm}

\newpage

\tableofcontents
	
\section{Introduction}
\noindent\textbf{Main objects.} Consider the Euclidean space $\mathbb{R}^n$, where $n\in\mathbb{N}$, and introduce the Laplace operator in the Cartesian coordinates $A_n(x)=-\sum_{i=1}^n\partial_{x_i}\partial_{x_i}$. Next, let $r>0$, then we can define the set of functions
\begin{equation}\label{35-02}
G_2(r)=-\frac{\ln(r)}{2\pi}\,\,\,\mbox{and}\,\,\,
G_n(r)=\frac{r^{2-n}}{(n-2)S_{n-1}}\,\,\,\mbox{for}\,\,\,n\neq2,
\end{equation}
where $S_{n-1}=2\pi^{n/2}/\Gamma(n/2)$ denotes the area of the unit sphere in $\mathbb{R}^n$. In particular, $S_0=2$. It is known, see section 3.1 in \cite{VZ}, that $G_n(\cdot)$ is a fundamental solution for the operator $A_n(x)$, which should be understood as a solution to the equation $A_n(x)G_n(|x|)=\delta^n(x)$ in the sense of generalized functions on the Schwartz class $\mathcal{S}(\mathbb{R}^n)$, see \cite{Gelfand-1964,Vladimirov-2002}, where $\delta^{n}(\cdot)$ is the corresponding $\delta$-functional.\\ 

\noindent\textbf{Motivation.} The notable position of the presented functions in theoretical physics is due to the fact, see \cite{29,35-DW, 30-1-1}, that they are the main approximation for Green's functions near the diagonal on smooth compact Riemannian manifolds. In particular, it is for this reason that they often arise in the study of perturbative decompositions for quantum field models \cite{3,7,10}. As is known, within the framework of this approach, non-linear combinations of Green's functions and their derivatives can appear, therefore, the regularization of Green's functions, which consists in deforming certain parts of it, plays a separate important role. This paper focuses on studying some important properties from a physical point of view for a special deformation of the functions $G_n(|\cdot|)$, see \cite{Iv-2024,Ivanov-2022,sk-b-20}, which can be obtained by applying the averaging operator twice.\\ 

\noindent\textbf{Averaging operator.}  Let us introduce the main elements of the averaging operator $O$. Let $\Lambda>0$ and a function $\omega(|\cdot|)\in C^{\infty}(\mathbb{R}^n,\mathbb{R})$ with the properties
\begin{equation}\label{35-04}
\mathrm{supp}(\omega)\subset[0,1/2],\,\,\,
\int_{\mathbb{R}^n}\mathrm{d}^nx\,\omega(|x|)=1.
\end{equation}
Then we will assume that the averaging operator $O_x^\Lambda$, where the index "$x$" denotes the variable for which averaging is performed, and the parameter $\Lambda$ is responsible for the averaging radius, acts on the function $\phi(\cdot)\in L_{1,\mathrm{loc}}(\mathbb{R}^n)$ according to the following rule
\begin{equation*}
O_x^\Lambda\phi(x)=\int_{\mathrm{B}_{1/2}}\mathrm{d}^ny\,\phi(x+y/\Lambda)\omega(|x|).
\end{equation*}
Here, $\mathrm{B}_{1/2}\subset\mathbb{R}^n$ denotes a closed ball of radius $1/2$ centered at the origin. In physics, the parameter $\Lambda$ is called regularizing, and the process of removing regularization is the transition to the limit of $\Lambda\to+\infty$. Indeed, in this case, $O_x^\Lambda\phi(x)\to\phi(x)$ due to the properties of \eqref{35-04}.\\

\noindent\textbf{Regularization.} Before defining the regularized (deformed) Green's function, we need to explain the appearance of double averaging in applications. It is known, see \cite{Bog-R}, that the coefficients of formal series in the framework of perturbative quantum field theory are obtained by calculating Gaussian integrals from polynomials, which ultimately boil down to repeated applications of Wick's theorem on pairings. Roughly speaking, the field functions responsible for fluctuations near a solution of the equation of motion are replaced in pairs by the corresponding Green's functions. In the main order, this can be written as follows
\begin{equation*}
\phi(x)\phi(y)\to G_n(x-y).
\end{equation*}
One of the recent approaches to the regularization of quantum field systems is the quasi-local probability averaging of fluctuation fields, see \cite{sksk}, that is, the transition $\phi(x)\to O_x^\Lambda\phi(x)$. Although the introduction of regularization using a smoothing integral operator is not original, see, for example, a cutoff in the momentum representation \cite{nkf-1,nkf-2,nkf-3} within the framework of the functional renormalization group or the method with higher covariant derivatives \cite{AA-2,29-st,Bakeyev-Slavnov}, the approach with averaging over a "small" neighborhood is new. Its systematic study was initiated in the study of the cubic model \cite{34} and significantly improved in subsequent works, see \cite{sksk,Iv-2024-1} and references therein. The word "quasi-locality" in this case is associated with the presence of a small area (a ball of radius $1/(2\Lambda)$) near a fixed point $x$ for averaging. After such a transformation, the Wick's theorem on pairing works according to the rule
\begin{equation*}
O_x^\Lambda\phi(x)O_y^\Lambda\phi(y)\to O_x^\Lambda O_y^\Lambda G_n(x-y)\equiv G_{n,\omega}^{\Lambda}(|x-y|)
\end{equation*}
and, thus, leads to the appearance of the double averaging. Next, by quasi-local averaging of the fundamental solution $G_n(|\cdot|)$, we name the transformation of the form
\begin{equation}\label{35-05}
G_n(|x|)\to
G_{n,\omega}^{\Lambda}(|x|)=\int_{\mathbb{R}^n}\mathrm{d}^ny\int_{\mathbb{R}^n}\mathrm{d}^nz\,
G_n(|x+y/\Lambda+z/\Lambda|)\omega(|y|)\omega(|z|).
\end{equation}
In this case, we call quasi-local averaging probabilistic if, in addition to the properties of \eqref{35-04}, $\omega(\cdot)\geqslant0$ is valid almost everywhere.\\

\noindent\textbf{About representations.} Note that it is sufficient to define the deformed function only for $\Lambda=1$, since the following large-scale transformations are valid for the initial functions \eqref{35-02}
\begin{equation}\label{35-06}
G_2(r/\Lambda)=G_2(r)+\frac{\ln(\Lambda)}{2\pi}\,\,\,\mbox{and}\,\,\,
G_n(r/\Lambda)=\Lambda^{2-n}G_n(r)\,\,\,\mbox{for}\,\,\,n\neq2,
\end{equation}
which are preserved when switching to the regularized objects. Therefore, further, assuming $\Lambda=1$, we will omit the corresponding index. It was previously shown, see Section 3 in \cite{sksk}, that the deformed function $G_{n,\omega}^{\Lambda}(|\cdot|)$ admits the following representation
\begin{equation}\label{35-07}
G_{n,\omega}(|x|)=G_{n}^{\mathbf{f}}(|x|)=
\mathbf{f}_n(|x|^2)+
\begin{cases}
\,\,G_n(1),&|x|\leqslant1; \\
G_n(|x|),&|x|>1,
\end{cases}
\end{equation}
where the auxiliary function $\mathbf{f}_n(\cdot)$ has the set of properties
\begin{equation}\label{35-08}
	\mathbf{f}_n(\cdot)\in C([0,+\infty),\mathbb{R}),\,\,\,
	\mathrm{supp}(\mathbf{f}_n)\subset[0,1],\,\,\,
	\mathbf{f}_n(1)=0.
\end{equation}

\noindent\textbf{Structure of the work.} This article consists of an introduction, three main sections and a conclusion. Here is a brief description of the sections and the corresponding results.\\

\noindent ---  Section \ref{35:sec:re}  contains the basic properties of the function \eqref{35-05} and kernels of the averaging operator, which are often used in practice, for example, in the framework of perturbative quantum field theory in the search for coefficients of renormalization constants. Theorem \ref{35-t-1} includes an explicit representation, see \eqref{35-01}, for double averaging of the function $G_n(|\cdot|)$ over spheres of arbitrary radius, as well as its basic properties such as smoothness and monotonicity intervals. Note that some intermediate results related to the theory of double averaging have appeared earlier, see for example section 1.5 in \cite{35-cit-3}, devoted to convolutions of $\delta$-functionals. Nevertheless, in the context of regularization, this topic is new, continuing the series of papers \cite{Ivanov-2022,Iv-2024}, therefore the theorem is useful from the point of view of further applications. Theorem \ref{35-t-2} is devoted to an example of a suitable class of kernels $\omega(|\cdot|)$, as well as to the derivation of representations for some special functionals.\\

\noindent ---  Section \ref{35:sec:pr} provides proofs for the statements made in the previous section.\\

\noindent ---  Section \ref{35:sec:sp} examines new examples that are important from the point of view of applications in the framework of perturbative quantum field theory and renormalization theory. The section consists of three parts. 
\begin{itemize}
	\item[(\textit{i})] Section \ref{sec:35:1} provides an example with double averaging over spheres of the same radius. This choice is very popular in scalar models, see \cite{Iv-2024-1} and the corresponding references in the introduction. Moreover, it represents an extreme case, namely, it minimizes the maximum value of the deformed Green's function, and thus leads to the extreme case of the renormalized mass.
	\item[(\textit{ii})] Section \ref{sec:35:2} provides an explicit calculation for a family of kernels in three-dimensional space. This result is important in the context of studying the properties of the sextic model \cite{29-3} with the cutoff in the coordinate representation, see \cite{Kh-2024,Kh-25}, since the free parameter continuously connects the extreme case with the limiting situation of regularization removal. This fact gives an additional degree of freedom when fixing the renormalization coefficients.
	\item[(\textit{iii})] Section \ref{sec:35:3} discusses a family of kernels in the two-dimensional case, which are an example of mixed regularization, that is, combining cutoffs in the coordinate and momentum representations. Using the example of the two-dimensional non-linear sigma model \cite{sig1} with the cutoff regularization, see \cite{sk-b-19,AIK-25,Ivanov-Akac}, it is shown that this approach leads to additional freedom, allowing us to make special-type functionals (functions $\theta_i$ from \cite{AIK-25}) equal to zero.
\end{itemize}

\noindent ---  Section \ref{35:sec:zak} contains concluding remarks and open questions.

\section{Main properties}\label{35:sec:re}
\begin{theorem}\label{35-t-1} Let $n\in\mathbb{N}$ and $x,y,z\in\mathbb{R}^n$. Let us introduce notation for the absolute values of $r=|x|$, $s=|y|$, and $t=|z|$. The symbol $\mathrm{S}^{n-1}$ denotes the unit $(n-1)$-dimensional sphere centered at the origin. Consider an integral of the form
\begin{equation}\label{35-20}
k_n(r,s,t)=
\int_{\mathrm{S}^{n-1}}\frac{\mathrm{d}^{n-1}\sigma(\hat{y})}{S_{n-1}}\int_{\mathrm{S}^{n-1}}\frac{\mathrm{d}^{n-1}\sigma(\hat{z})}{S_{n-1}}\,
G_n(|x+s\hat{y}+t\hat{z}|),
\end{equation}
where $\hat{y}=y/|y|$, $\hat{z}=z/|z|$, and also $\mathrm{d}^{n-1}\sigma(\hat{y})$ denotes the standard measure on the sphere, which is normalized by the number $S_{n-1}$ and, at $n=1$, degenerates into the summation over the set $\{\pm1\}$. Taking into account all the above, the relation holds
\begin{equation}\label{35-01}
k_n(r,s,t)=
\begin{cases}
G_n(\mathrm{max}(t,s)),	& r<|t-s|;  \\
G_n(t+s)+g_n(r,s,t),
	& |t-s|\leqslant r\leqslant|t+s|; \\
G_n(r),	& r>|t+s|,
\end{cases}
\end{equation}
where
\begin{equation}\label{35-03}
g_n(r,s,t)=
\frac{2^{n-2}S_{n-2}}{S_{n-1}^2}
\int_r^{t+s}\mathrm{d}u\,u^{1-n}\int_0^{(u^2-(t-s)^2)/(4ts)}\mathrm{d}p\,p^{\frac{n-3}{2}}(1-p)^{\frac{n-3}{2}}
\end{equation}
for $n>1$ and $g_1(r,s,t)=(r-t-s)/2$. In this case, the following four sets of properties are valid.\\

\noindent1) The function $k_n(r,s,t)$ is continuous, symmetric, and decreasing with respect to each argument on the set $\mathbb{R}_{3}=\mathbb{R}_+^3\setminus\{0,0,0\}$ for all $n\in\mathbb{N}$. Moreover, the ratio is correct
%	\begin{align}
\begin{equation}
	\label{35-21}
	k_n(r,s,t)>0\,\,\,\mbox{for}\,\,\,n\geqslant3\,\,\,\mbox{and}\,\,\,(r,s,t)\in\mathbb{R}_3.
\end{equation}
%\\\label{35-22}
%	k_2(r,s,t)\geqslant0\,\,\,&\mbox{для}\,\,\,(r,s,t)\in\big(\mathbb{R}_+\times[0,1/2]\times[0,1/2]\big)\cap\mathbb{R}_3.
%\end{align}
\noindent2) For fixed  $(s,t)\in\mathbb{R}_+^2\setminus\{0,0\}$ and $n\geqslant3$ the function $k_n(\cdot,s,t)$ reaches its maximum at $\mathbb{R}_+$ at point $0$, which is
\begin{equation}
k_n(0,s,t)=G_n(\mathrm{max}(t,s)).
\end{equation}

\noindent3) Let $(s,t)\in\mathbb{R}_+^2\setminus\{0,0\}$ are fixed and $n\geqslant2$, then, with respect to the variable $r$, the function $\partial_rk_n(r,s,t)$ is decreasing and continuous on $\mathbb{R}_+$. In addition, the relations are correct
		\begin{align}\label{35-23}
		\partial_rk_n(r,s,t)=0\,\,\,&\mbox{for}\,\,\,r\leqslant|t-s|,
		\\\label{35-24}
		\partial_rk_n(r,s,t)<0\,\,\,&\mbox{for}\,\,\,r>|t-s|.
	\end{align}
In the case of $n=1$, the function has discontinuities of the type "jump" for $r=|t-s|$ and $r=|t+s|$.\\

\noindent4) The function $A_n(x)k_n(r,s,t)$ is symmetric on $\mathbb{R}_3$. At the same time, for $n\geqslant4$ the function is continuous on $\mathbb{R}_3$, and for $n=3,2$ and all fixed $(s,t)\in\mathbb{R}_+^2\setminus\{0,0\}$ the function has break points\footnote{See Paragraph 2 of Chapter 2 in \cite{34-f-1}.} in $r\in\{|t-s|,|t+s|\}$, of the first kind for $n=3$ and of the second kind for $n=2$, and is continuous on the set $\mathbb{R}_+\setminus\{|t-s|,|t+s|\}$. 
In addition, for $n\geqslant2$, the relations are correct
\begin{align}\label{35-25}
	A_n(x)k_n(r,s,t)=0\,\,\,&\mbox{for}\,\,\,r\in\mathbb{R}_+\setminus[|t-s|,|t+s|],
	\\\label{35-26}
	A_n(x)k_n(r,s,t)<0\,\,\,&\mbox{for}\,\,\,r\in(|t-s|,|t+s|).
\end{align}
\end{theorem}

\begin{theorem}\label{35-t-2} Let the assumptions from Theorem \ref{35-t-1} be fulfilled, $\Omega=(0,1/2)$, $n\in\mathbb{N}$, and a function\footnote{Note that the multiplier $t^{-(n-\delta_{1n})/2+\nu_n}$ originally appeared for convenience reasons when describing a class of valid kernels for which the transition from \eqref{35-e-2} to \eqref{35-e-3} is possible.} is defined as $\omega(t)=t^{-(n-\delta_{1n})/2+\nu_n}w(t)$, where the parameter $\nu_n>0$ and
\begin{equation*}
w\in L_{1,+}(\Omega)=\{u\in L_1(\Omega): u\geqslant0\,\,\mbox{almost everywhere}\}.
\end{equation*}
Then the deformed function \eqref{35-05} allows a representation of the form
\begin{equation}\label{35-e-12}
G_{n,\omega}(r)=S^2_{n-1}\int_0^{1/2}\mathrm{d}s\int_0^{1/2}\mathrm{d}t\,s^{n-1}t^{n-1}
\omega(s)\omega(t)k_n(r,s,t),
\end{equation}
where $x\in\mathbb{R}^n$ and $|x|=r$, and the following four statements are true.\\
	
\noindent1) The function $G_{n,\omega}(\cdot)$ is continuous and decreasing on $\mathbb{R}_+$. At the same time, for $n\geqslant3$, it is bounded from below by zero, that is, $0<G_{n,\omega}(r)$ for all $r\in\mathbb{R}_+$. In addition, for $n\neq2$, formula \eqref{35-e-12} is valid for the extreme case $\nu_n\to+0$. The maximum is reached at the point $r=0$ and is equal to
\begin{equation}\label{35-e-2}
G_{n,\omega}(0)=2S^2_{n-1}\int_0^{1/2}\mathrm{d}s\,s^{n-1}G_n(s)\omega(s)\int_0^s\mathrm{d}t\,t^{n-1}\omega(t)\equiv\mathrm{I}_1[w].
\end{equation}
\noindent2) For $\nu_n>0$ the value of $G_{n,\omega}(0)$ can be rewritten in the equivalent form
\begin{equation}\label{35-e-3}
G_{n,\omega}(0)=G_n(1/2)+S_{n-1}\int_0^{1/2}\mathrm{d}s\,s^{1-n}\bigg(\int_0^s\mathrm{d}t\,t^{n-1}\omega(t)\bigg)^2\equiv\mathrm{I}_2[w].
\end{equation}
Moreover, in the case of $n=1$, the formula is also true for $\nu_1=0$.\\

\noindent3) If $\nu_n\geqslant1/2$ for all $n\geqslant2$ or $\nu_1\geqslant0$, then the function $\partial_rG_{n,\omega}(r)$ is continuous on $\mathbb{R}_+$ and has two properties
\begin{equation}\label{35-e-9}
\partial_rG_{n,\omega}(r)\Big|_{r=0}=0
,\,\,\,
\partial_rG_{n,\omega}(r)\Big|_{r\geqslant1}=-\frac{r^{1-n}}{S_{n-1}}.
\end{equation}

\noindent4) If $\nu_n\geqslant1/2$ for all $n\geqslant2$ or $\nu_1\geqslant0$, and also $w\in L_2(\Omega)$, then the function $A_n(x)G_{n,\omega}(r)$ is continuous on $\mathbb{R}_+$ and has two properties
\begin{equation}\label{35-e-10}
	A_n(x)G_{n,\omega}(r)\Big|_{r=0}=\big|\big||\cdot|^{\nu_n-(1-\delta_{1n})/2}w(|\cdot|)\big|\big|_{L_2(\Omega)}^2
	,\,\,\,
	A_n(x)G_{n,\omega}(r)\Big|_{r\geqslant1}=0.
\end{equation}
\end{theorem}

\begin{cor}\label{35-cor-1}
Let $w\in L_{1,+}(\Omega)$ from Theorem \ref{35-t-2}. Let there be such a number $1/2>\epsilon>0$, that $w(r)=0$ for almost all $r\in[0,\epsilon]$. Then the first three points of Theorem \ref{35-t-2} are true for all $\nu_n\in\mathbb{R}$. Besides, the following estimate is valid
\begin{equation}\label{35-h-1}
G_n(1/2)\leqslant G_{n,\omega}(0)\leqslant G_n(\epsilon).	
\end{equation}
\end{cor}

\section{Proofs}\label{35:sec:pr}

\noindent\textbf{Proof of Theorem \ref{35-t-1}.}
Let us first consider the case of $n>2$. Let us apply the Laplace operator $A_n(x)$ to the kernel $k_n(r,s,t)$, then we get
\begin{equation}\label{35-10}
A_n(x)k_n(r,s,t)=\int_{\mathrm{S}^{n-1}}\frac{\mathrm{d}^{n-1}\sigma(\hat{y})}{S_{n-1}}\int_{\mathrm{S}^{n-1}}\frac{\mathrm{d}^{n-1}\sigma(\hat{z})}{S_{n-1}}\,
\delta^{n}(x+s\hat{y}+t\hat{z}),
\end{equation}
where the equation  $A_n(x)G_n(|x|)=\delta^n(x)$ was used. Recall that $\delta^{n}(\cdot)$ is the corresponding $\delta$-functional on the class $\mathcal{S}(\mathbb{R}^n)$, see \cite{Gelfand-1964,Vladimirov-2002}. Note that this function has been studied before, see pages 381 and 404 in the monograph \cite{35-cit-1} and page 585 in \cite{35-cit-2}. Nevertheless, a short summary of the conclusion is useful from the point of view of general understanding. Let us use the relation for the internal integral, see formula (14) in \cite{Iv-2024},
\begin{equation}\label{35-11}
\int_{\mathrm{S}^{n-1}}\mathrm{d}^{n-1}\sigma(\hat{z})\,
\delta^{n}(x+s\hat{y}+t\hat{z})=t^{1-n}\delta^1(|x+s\hat{y}|-t),
\end{equation}
and let us move to hyperspherical coordinates, then writing out the explicit formula for $\mathrm{d}^{n-1}\sigma(\hat{z})$ using the angular coordinates and additionally calculating the integral over the $(n-2)$-dimensional hypersphere, we obtain a representation for the right side of \eqref{35-10} in the form
\begin{equation}\label{35-12}
\frac{t^{1-n}S_{n-2}}{S_{n-1}^2}
\int_0^\pi\mathrm{d}\phi\,\sin^{n-2}(\phi)\,
\delta^1\big(v(\phi)\big),
\end{equation}
where the auxiliary function was defined
\begin{equation}\label{35-13}
v(\phi)=\big(r^2+2rs\cos(\phi)+s^2\big)^{1/2}-t.
\end{equation}
The search for zeros for this function reduces to solving the equation
\begin{equation}\label{35-14}
\cos(\phi_c)=\frac{t^2-r^2-s^2}{2rs}.
\end{equation}
Considering the parameters $t$ and $s$ fixed, the equation has a nonzero solution in the interval $[0,\pi]$ only if the condition is fulfilled
\begin{equation}\label{35-15}
|t-s|\leqslant r\leqslant|t+s|.
\end{equation}
Thus, calculating the value of the first derivative
\begin{equation}\label{35-16}
\dot{v}(\phi_c)=-\frac{rs\sin(\phi_c)}{v(\phi_c)+t}=-\frac{rs}{t}\sqrt{1-\bigg(\frac{t^2-r^2-s^2}{2rs}\bigg)^2}
\end{equation}
and replacing the variables in the functional
\begin{equation}\label{35-17}
\delta^1\big(v(\phi)\big)=\frac{\delta^1(\phi-\phi_c)}{|\dot{v}(\phi_c)|},
\end{equation}
we obtain the following result for integral \eqref{35-10}
\begin{equation}\label{35-18}
\theta(r-|t-s|)\theta(|t+s|-r)
\frac{2S_{n-2}}{S_{n-1}^2}
(2tsr)^{2-n}
\bigg(\Big(r^2-(t-s)^2\Big)\Big((t+s)^2-r^2)\Big)\bigg)^{\frac{n-3}{2}}
,
\end{equation}
where $\theta(\cdot)$ denotes the Heaviside (step) function. Next, representing the Laplace operator as follows $-r^{1-n}\partial_rr^{n-1}\partial_r$ and integrating over the variable $r$, taking into account continuous gluing at points  $r=|t-s|$ and $r=|t+s|$ and equalities $k_n(r,s,t)=G_n(r)$ for $r>|t+s|$, we obtain formula \eqref{35-01}. Note that, knowing the answer, we can simply check the fulfillment of all the properties and conditions presented in the formulation of the theorem. Indeed, $k_n(r,s,t)$ is continuous, since the function from \eqref{35-03} satisfies the conditions
\begin{equation}\label{35-19}
g_n(|t-s|,s,t)=G_n(\mathrm{max}(t,s))-G_n(t+s),\,\,\,
g_n(|t+s|,s,t)=0.
\end{equation}
The second follows from the direct substitution of $r=t+s$, while the first is obtained in several steps. To check, select $r=|t-s|$ in \eqref{35-03} and replace the variable with $u^2=4tsv+(t-s)^2$, then we get
\begin{equation*}
\frac{2^{n-1}S_{n-2}ts}{S_{n-1}^2}
\int_{0}^{1}\mathrm{d}v\,\Big(4tsv+(t-s)^2\Big)^{-\frac{n}{2}}\int_{0}^{v}\mathrm{d}p\,p^{\frac{n-3}{2}}(1-p)^{\frac{n-3}{2}}.
\end{equation*}
Next, we change the order of integration and explicitly calculate the inner integral. When calculating, for convenience, we can assume that $s\neq t$, understanding the situation $s=t$, which was studied in \cite{Iv-2024}, as limit case. We have
\begin{equation*}
-\frac{2^{n-2}S_{n-2}}{S_{n-1}^2}\frac{|t+s|^{2-n}}{n-2}\mathrm{B}\bigg(\frac{n-1}{2},\frac{n-1}{2}\bigg)+
\frac{2^{n-2}S_{n-2}|t-s|^{2-n}}{S_{n-1}^2(n-2)}\mathrm{P}_n(t,s),
\end{equation*}
where $\mathrm{B}(\,\cdot\,,\,\cdot\,)$ is the Euler's beta-function and $\mathrm{P}_n(t,s)$ is defined by the equality
\begin{equation*}
P_n(t,s)=
\int_{0}^{1}\mathrm{d}p\,p^{\frac{n-3}{2}}(1-p)^{\frac{n-3}{2}}
\bigg(1-p\bigg(-\frac{4ts}{(t-s)^2}\bigg)\bigg)^{\frac{2-n}{2}}.
\end{equation*}
Then we use the equalities, see paragraphs 8.383, 9.121, and 9.131 in \cite{GR},
\begin{equation*}
\frac{2^{n-2}S_{n-2}}{S_{n-1}}\mathrm{B}\bigg(\frac{n-1}{2},\frac{n-1}{2}\bigg)=1,
\end{equation*}
\begin{equation*}
\mathrm{P}_n(t,s)=\frac{\Gamma^2\big((n-1)/2\big)}{\Gamma(n-1)}
{}_{2}\mathrm{F}_{1}\bigg(\frac{n-2}{2},\frac{n-1}{2},n-1;-\frac{4ts}{(t-s)^2}\bigg),
\end{equation*}
\begin{equation*}
{}_{2}\mathrm{F}_{1}\bigg(\frac{n-2}{2},\frac{n-1}{2},n-1;-z\bigg)=2^{n-2}
\Big(1+\sqrt{1-z}\Big)^{2-n},
\end{equation*}
\begin{equation*}
|t-s|+t+s=2\max(t,s),
\end{equation*}
where $z>0$. Finally, after substituting the mentioned relations and reducing identical terms, we come to the first relation from \eqref{35-19}. Thus, the continuity of $k_n(r,s,t)$ over the variable $r$ is shown. In this case, the function is strictly positive. Calculate the first derivative for \eqref{35-01}
\begin{equation}\label{35-e-4}
\partial_r k_n(r,s,t)=-\frac{r^{1-n}}{S_{n-1}}
\begin{cases}
	0,	& r<|t-s|;  \\
	h_n(r,s,t),
	& |t-s|\leqslant r\leqslant|t+s|; \\
	1,	& r>|t+s|,
\end{cases}
\end{equation}
where
\begin{equation}\label{35-e-5}
h_n(r;s,t)=\frac{2^{n-2}S_{n-2}}{S_{n-1}}
\int_0^{(r^2-(t-s)^2)/(4ts)}\mathrm{d}p\,p^{\frac{n-3}{2}}(1-p)^{\frac{n-3}{2}}.
\end{equation}
Let us verify the continuity of the function. For $r=|t-s|$ the equality $h_n(|t-s|;s,t)=0$ is valid. For $r=t+s$, using the above mentioned relations, we obtain in a similar way $h_n(t+s;s,t)=1$. Therefore, the first derivative is a continuous non-positive decreasing function. At the same time, it is zero only when $r\leqslant|t-s|$. Next, applying the operator $-r^{1-n}\partial_rr^{n-1}$, we make sure that it is consistent with the presentation \eqref{35-18}. Note that the symmetry of the functions $k_n(r,s,t)$ and $A_n(x)k_n(r,s,t)$ follows from the fact that they depend on the absolute value of $|x|$. Indeed, in this case, we can add averaging over the variable $\hat{x}=x/r$, which lead to a symmetric form in formulas \eqref{35-20} and \eqref{35-10}. Thus, taking into account the mentioned part of the proof, we obtain properties 1)--4) for $n>2$. Following similar steps, we verify the validity of relation \eqref{35-01} for $n=1,2$ and the corresponding properties. The theorem is proved.\\

\noindent\textbf{Proof of Theorem \ref{35-t-2}.} \\

\noindent\underline{The first point.} Let $t\in\Omega$ and $\omega(t)=t^{-(n-\delta_{1n})/2+\nu_n}w(t)$, where $w\in L_{1,+}(\Omega)$ and $\nu_n\geqslant0$. Using property 1 from Theorem \ref{35-t-1} and the non-negativity of the kernel $\omega(\cdot)$ for almost all values of the argument, we obtain decreasing of the function $G_{n,\omega}(\cdot)$ on $\mathbb{R}_+$. In this case, the lower boundedness for $n\geqslant3$ follows from the boundedness of the function $k_n$. Next, using property 2 of Theorem \ref{35-t-1}, we obtain explicit formula \eqref{35-e-2} for $G_{n,\omega}(0)$. In this case, the integral is finite, since the estimates are valid
\begin{equation*}
\int_0^s\mathrm{d}t\,t^{n-1}\omega(t)=
\int_0^s\mathrm{d}t\,t^{(n+\delta_{1n})/2-1+\nu_n}w(t)\leqslant s^{(n+\delta_{1n})/2-1+\nu_n}||w||_{L_1},
\end{equation*}
\begin{equation*}
\Big|\mathrm{I}_1[w]\Big|\leqslant2S_{n-1}^2||w||_{L_1}^2\max_{s\in[0,1/2]}\Big(s^{n+\delta_{1n}-2+2\nu_n}G_n(s)\Big).
\end{equation*}
At the same time, if $n\neq2$, then the transition $\nu_n\to +0$ is valid, since the function in parentheses remains continuous at $[0,1/2]$. The first point is proved.\\

\noindent\underline{The second point.} Note that if $\omega\in C_c^{\infty}(\Omega)$ then in $\mathrm{I}_1[w]$ we can integrate in parts, eventually getting the stated answer \eqref{35-e-3}. In this case, the restriction of $\nu_n>0$ on the parameter leads to the fact that the function $t^{2\nu_n-1}$ belongs to $L_1(\Omega)$. Turning to generalization, let us use the density of the set $C_c^{\infty}(\Omega)$ in $L_1(\Omega)$, see Corollary 4.23 in \cite{35-hm}. Let the sequence $\{w_i\in C_c^{\infty}(\Omega)\}_{i=1}^{+\infty}$ tend to $w$ in $L_1(\Omega)$, that is, for any $\varepsilon>0$ there is such a number $N\in\mathbb{N}$, that $||w_i-w||_{L_1}<\varepsilon$ for all $i>N$. Moreover, we can additionally assume that $||w_i||_{L_1}<(1+i^{-1})||w||_{L_1}$. In this case, the estimates are valid for $n\geqslant3$ and $i>N$
\begin{equation*}
\big|\mathrm{I}_1[w]-\mathrm{I}_2[w]\big|\leqslant
\big|\mathrm{I}_1[w]-\mathrm{I}_1[w_i]\big|+
\big|\mathrm{I}_1[w_i]-\mathrm{I}_2[w]\big|=
\big|\mathrm{I}_1[w]-\mathrm{I}_1[w_i]\big|+
\big|\mathrm{I}_2[w_i]-\mathrm{I}_2[w]\big|,
\end{equation*}
\begin{equation*}
\big|\mathrm{I}_1[w]-\mathrm{I}_1[w_i]\big|\leqslant\frac{2S_{n-1}}{n-2}\big(||w||_{L_1}+||w_i||_{L_1}\big)||w_i-w||_{L_1}\leqslant\varepsilon
\frac{6 S_{n-1}}{n-2}||w||_{L_1},
\end{equation*}
\begin{equation*}
\big|\mathrm{I}_2[w]-\mathrm{I}_2[w_i]\big|\leqslant 2^{-2\nu_n}\frac{S_{n-1}}{2\nu_n}\big(||w||_{L_1}+||w_i||_{L_1}\big)||w_i-w||_{L_1}\leqslant\varepsilon\frac{3S_{n-1}}{2\nu_n}||w||_{L_1}.
\end{equation*}
Hence, we get
\begin{equation*}
	\big|\mathrm{I}_1[w]-\mathrm{I}_2[w]\big|\leqslant
3\varepsilon||w||_{L_1}S_{n-1}\frac{4\nu_n-2+n}{2(n-2)\nu_n}.
\end{equation*}
Since the left part does not depend on $\varepsilon$ and the index $i$, we obtain the desired equality of integrals. The cases of $n=1,2$ are treated similarly. In this case, for $n=1$, the case of $\nu_1=0$ is valid as well. The second point is proved.\\

\noindent\underline{The third point.} Let $n\geqslant2$. Let us consider the first derivative $\partial_rG_{n,\omega}(r)$. We immediately note the equality $\partial_rG_{n,\omega}(r)=\partial_rG_{n}(r)$ in the region $r\geqslant1$, which follows from the explicit formula \eqref{35-01}. Next, let us use representation \eqref{35-e-4} and study the function $\hat{h}_n(r)=r^{1-n}h_n(r,s,t)$ in more detail with fixed $(s,t)\in\overline{\Omega}\times\overline{\Omega}\setminus\{u\times u:u\in\Omega\}$ in the interval $|t-s|\leqslant r\leqslant t+s$. To do this, first, in integral \eqref{35-e-5}, we scale the variable $p\to p\varkappa$, where $\varkappa=(r^2-(t-s)^2))/(4ts)$, then we get
\begin{equation*}
\hat{h}_n(r;s,t)=\frac{2^{n-2}S_{n-2}}{S_{n-1}}
(\varkappa/r^2)^{\frac{n-1}{2}}
\int_0^1\mathrm{d}p\,p^{\frac{n-3}{2}}(1-p\varkappa)^{\frac{n-3}{2}}.
\end{equation*}
Note that for the specified range of values of $r$, the parameter $\varkappa\in[0,1]$. Therefore, the estimates are valid
\begin{equation*}
\varkappa/r^2=\frac{1}{4ts}\bigg(1-\frac{(t-s)^2}{r^2}\bigg)\leqslant\frac{1}{(t+s)^2},\,\,\,
1-p\leqslant1-p\varkappa\leqslant1.
\end{equation*}
Next, consider the case of $n\geqslant3$. Using the both upper bounds, we obtain
\begin{equation*}
	\hat{h}_n(r;s,t)\leqslant\frac{2^{n-2}S_{n-2}}{S_{n-1}(t+s)^{n-1}}
	\int_0^1\mathrm{d}p\,p^{\frac{n-3}{2}}\leqslant\frac{2^{n-2}S_{n-2}}{S_{n-1}(t+s)^{n-1}}.
\end{equation*}
In the case $n=2$, the lower bound should be used for the integral expression, then we get
\begin{equation*}
	\hat{h}_2(r;s,t)\leqslant\frac{S_{0}}{S_{1}(t+s)}
	\int_0^1\mathrm{d}p\,p^{-\frac{1}{2}}(1-p)^{-\frac{1}{2}}=
	\frac{S_{0}\pi}{S_{1}(t+s)}.
\end{equation*}
Thus, defining the auxiliary value
\begin{equation*}
C_n=\frac{\pi2^{n-2}S_{n-2}}{S_{n-1}}\geqslant1,
\end{equation*}
and noticing that in the range $r\in[t+s,1]$ the inequality $C_n|t+s|^{1-n}\geqslant r^{1-n}$ holds, we get an estimate of the form
\begin{equation}\label{35-e-6}
	|\partial_r k_n(r,s,t)|\leqslant\frac{C_n}{S_{n-1}}\Bigg(
	\begin{cases}
		0,	& |t-s|>r\geqslant0;  \\
		|t+s|^{1-n},
		& |t-s|\leqslant r<1,
	\end{cases}\Bigg)\leqslant\frac{C_n|t+s|^{1-n}}{S_{n-1}}.
\end{equation}
Using the second inequality and integrating, for $r\in[0,1)$, we obtain an inequality of the form
\begin{equation}\label{35-e-7}
\Big|\partial_rG_{n,\omega}(r)\Big|\leqslant C_nS_{n-1}\int_0^{1/2}\mathrm{d}s\int_0^{1/2}\mathrm{d}t\,
\frac{(st)^{n/2-1+\nu_n}}{(s+t)^{n-1}}w(s)w(t)\leqslant 2^{(1-n)/2}C_nS_{n-1}||w||_{L_1}^2,
\end{equation}
where the inequalities $2st<(s+t)^2$ and $(st)^{\nu_n-1/2}\leqslant1$ were used. Thus, the function $\partial_rG_{n,\omega}(r)$ is bounded and continuous for $\nu_n\geqslant1/2$. Let us show that it is zero at $r=0$. Note that when $r\to+0$, after integrating the first inequality from \eqref{35-e-6}, the estimate follows
\begin{align*}
\Big|\partial_rG_{n,\omega}(r)\Big|&\leqslant 2C_nS_{n-1}\bigg(\int_r^{1/2}\int_{s-r}^{s}+\int_0^{r}\int_{0}^{s}\bigg)\mathrm{d}s\mathrm{d}t\,
\frac{(st)^{n/2-1+\nu_n}}{(s+t)^{n-1}}w(s)w(t)\\
&\leqslant 2^{(5-n)/2}C_nS_{n-1}||w||_{L_1}\max_{s\in[r,1/2]}\bigg(\int_{s-r}^{s}\mathrm{d}t\,w(t)\bigg),
\end{align*}
the right-hand side of which tends to zero when $r\to+0$. Hence, we get $\partial_rG_{n,\omega}(r)\big|_{r=0}=0$. Similarly, taking into account explicit form \eqref{35-03}, the case $n=1$ can be analyzed. The third point is proved.\\

\noindent\underline{The fourth point.} Let us move on to the consideration of $A_n(x)G_{n,\omega}(r)$. In this case, it is convenient to use the representation of $\eqref{35-05}$, then for $\omega(|\cdot|)\in L_2(\mathrm{B}_{1/2})$, where $\mathrm{B}_{1/2}=\{x\in\mathbb{R}^n:|x|<1/2\}$, the equality holds
\begin{equation}\label{35-e-8}
A_n(x)G_{n,\omega}(x)=\int_{\mathbb{R}^n}\mathrm{d}^ny\,\omega(|x-y|)\omega(|y|),
\end{equation}
In this case, $L_2(\mathrm{B}_{1/2})$ is understood as a natural subset of $L_2(\mathbb{R}^n)$, that is, $\omega(t)=0$ for all $t>1/2$. Next, we observe that the norm in $L_2(\mathrm{B}_{1/2})$, taking into account the spherical symmetry of the function, can be rewritten as follows
\begin{align}\label{35-e-11}
\big|\big|\omega(|\cdot|)\big|\big|_{L_2(\mathrm{B}_{1/2})}^2&=S_{n-1}\int_0^{1/2}\mathrm{d}s\,
\Big(s^{\nu-(1-\delta_{1n})/2}w(s)\Big)^2\\\nonumber&=S_{n-1}
\big|\big||\cdot|^{\nu-(1-\delta_{1n})/2}w(|\cdot|)\big|\big|_{L_2(\Omega)}^2.
\end{align}
Then, taking into account the restrictions for the supports, we get the properties from \eqref{35-e-10}. Consequently, the fourth point has also been verified, and the theorem is completely proved.\\

\noindent\textbf{Proof of Corollary \ref{35-cor-1}.} Since $\omega(r)=0$ for almost all $r\in[0,\epsilon]$, then $t^a\omega(t)$ belongs to the class $L_{1,+}(\Omega)$. Therefore, by choosing the parameter $a$ appropriately, it is possible to satisfy the conditions of the first three points of Theorem \ref{35-t-2}. Relation \eqref{35-h-1} follows from the application of formula \eqref{35-e-3} using an estimate of the form
\begin{equation*}
S_{n-1}\int_0^{1/2}\mathrm{d}s\,s^{1-n}\bigg(\int_0^s\mathrm{d}t\,t^{n-1}\omega(t)\bigg)^2\leqslant
\frac{||\omega||_{L_1(\mathrm{B}_{1/2})}^2}{S_{n-1}}
\int_\epsilon^{1/2}\mathrm{d}s\,s^{1-n}=G_n(\epsilon)-G_n(1/2),
\end{equation*}
where definition \eqref{35-02} and condition \eqref{35-04} were taken into account. Corollary \ref{35-cor-1} is proved.

\section{New examples}\label{35:sec:sp}

\subsection{The case $t=s=1/2$.}
\label{sec:35:1}

Let the notations from Theorem \ref{35-t-1} be true. Substitute the parameter values into formulas \eqref{35-01} and \eqref{35-03}, then we get
\begin{equation}\label{35-22}
	k_n(r,1/2,1/2)=
	\begin{cases}
		G_n(1)+g_n(r,1/2,1/2),
		&r\leqslant1; \\
		G_n(r),	& r>1,
	\end{cases}
\end{equation}
where
\begin{equation}\label{35-27}
	g_n(r,1/2,1/2)=
	\frac{2^{n-2}S_{n-2}}{S_{n-1}^2}
	\int_r^{1}\mathrm{d}u\,u^{1-n}\int_0^{u^2}\mathrm{d}p\,p^{\frac{n-3}{2}}(1-p)^{\frac{n-3}{2}}
\end{equation}
for $n>1$ and $g_1(r,1/2,1/2)=(r-1)/2$. This choice of parameters corresponds to the kernel \eqref{35-04} of a special type $\omega(|x|)\to\hat{\omega}(|x|)=2^{n-1}\delta^1(|x|-1/2+\epsilon)/S_{n-1}$ when $\epsilon\to+0$, which should be understood in the sense of generalized functions, see \cite{Gelfand-1964,Vladimirov-2002}. Normalization can be checked by direct substitution and switching to spherical coordinates. Here, the auxiliary regularization with the parameter $\epsilon$ reveals the ambiguity at the boundary of the domain, since integration occurs over $|x|\in[0,1/2]$. Note that this choice of kernel does not fit into Theorem \ref{35-t-2}, since the generalized function $\hat{\omega}(|\cdot|)$ is not contained in the set $L_{1,+}(\Omega)$. Nevertheless, such a variant exists and can be obtained as a limit of functions from $L_{1,+}(\Omega)$. Indeed, let us choose\footnote{Note that this is not the only choice. As an example, we can take any $\delta$-shaped sequence on the segment $[0,1/2]$ converging to the functional $\hat{\omega}(|x|)$.} a sequence of non-negative functions of the form 
\begin{equation}\label{35-40}
\hat{\omega}_k(r)=\frac{r^{1-n}}{S_{n-1}}
\begin{cases}
	0,
	&0\leqslant r<1/2-1/k; \\
	k,	& 1/2-1/k\leqslant r\leqslant1/2;\\
	0,
	&1/2<r. 
\end{cases}
\end{equation}
It is clear that $\hat{\omega}_k(|\cdot|)\to\hat{\omega}(|\cdot|)$ for $k\to+\infty$ in the sense of generalized functions on $C(\mathbb{R}_+)$, and, moreover, the normalization is valid
\begin{equation*}
\int_{\mathbb{R}^n}\mathrm{d}^nx\,\hat{\omega}_k(r)=1.
\end{equation*}
This example shows that the class of kernels is broader than the one described in Theorem \ref{35-t-2}. However, the found representative of a broader class can be obtained by appropriate limit transition.
Thus, we get
\begin{equation}\label{35-28}
G_{n,\hat{\omega}}(|x|)=\int_{\mathbb{R}_+}\mathrm{d}s\int_{\mathbb{R}_+}\mathrm{d}t\,s^{n-1}t^{n-1}2^{2n-4}
\delta^1(t-1/2)\delta^1(s-1/2)k_n(r,s,t)
=k_n(r,1/2,1/2)
\end{equation}
and
\begin{equation}\label{35-29}
\mathbf{f}_n(|x|^2)=G_{n,\hat{\omega}}(|x|)-\begin{cases}
	G_n(1),
	&r\leqslant1; \\
	G_n(r),	& r>1,
\end{cases}=g_n(r,1/2,1/2)\theta(1-r).
\end{equation}
Such an example of regularization has been discussed in recent papers. In particular, see \cite{sk-b-20}, it was shown that $\mathbf{f}_n(|x|^2)$ can be represented by a finite sum of elementary functions for any $n\in\mathbb{N}$. Explicit examples for $n\in\{3,\ldots,6\}$ are given in Corollary 1 of \cite{Iv-2024}. In conclusion, we note that the kernel $\hat{\omega}(|\cdot|)$ gives the minimum value for the regularized function $G_{n,\omega}(\cdot)$ from \eqref{35-e-3} at zero, which can be achieved, for example, by limit transition. Indeed, let $\omega=\hat{\omega}_k$, then
\begin{equation*}
G_{n,\hat{\omega}_k}(0)=G_{n}(1/2)+
\frac{1}{S_{n-1}}\int_{1/2-1/k}^{1/2}\mathrm{d}s\,s^{1-n}\Big(s-1/2+1/k\Big)^2
\xrightarrow{k\to+\infty}
G_{n,\hat{\omega}}(0)=G_{n}(1/2).
\end{equation*}
Intuitively, this result is highly expected. Indeed, the function $G_n(\cdot)$ is decreasing in the specified dimensions. Therefore, the minimum value of the averaged function should be achieved on kernels whose density is concentrated at the upper boundary of the segment $[0,1/2]$. This is exactly the option that is achieved by averaging over the sphere.

\subsection{The case $n=3$.}
\label{sec:35:2}

Substitute the index value, then $S_2=4\pi$, $S_1=2\pi$. In this case, formula \eqref{35-03} can be written out explicitly, but it is not convenient to calculate the integral with such a function, so we use \eqref{35-18}, then we get
\begin{equation*}
A_3(x)k_3(r,s,t)=\theta(r-|t-s|)\theta(|t+s|-r)
\frac{1}{8\pi tsr},
\end{equation*}
and, as a consequence, we get the formula
\begin{equation*}
A_3(x)G_{3,\omega}(r)
=16\pi^2\int_0^{1/2}\mathrm{d}s\int_0^{1/2}\mathrm{d}t\,s^2\omega(s)t^2\omega(t)A_3(x)k_3(r,s,t).
\end{equation*}
Note that for each fixed $r>0$, the integration domain is defined by the following conditions
\begin{equation*}
|t-s|\leqslant r,\,\,\,|t+s|\geqslant r,\,\,\,s,t\in[0,1/2],
\end{equation*}
therefore, the integral can be represented as follows
\begin{equation*}
A_3(x)G_{3,\omega}(r)=\frac{1}{4\pi}
	\begin{cases}
		A_3(x)R_1(r),	& r<1/2;  \\
		A_3(x)R_2(r),
		& 1/2\leqslant r\leqslant1; \\
		0,	& r>1,
	\end{cases}
\end{equation*}
where
\begin{equation*}
A_3(x)R_1(r)=\frac{8\pi^2}{r}\bigg(\int_{0}^{1/2}\int_{0}^{1/2}-2\int_{r}^{1/2}\int_{0}^{s-r}-\int_{0}^{r}\int_{0}^{r-s}\bigg)\mathrm{d}s\mathrm{d}t\,s\omega(s)t\omega(t),
\end{equation*}
\begin{equation*}
A_3(x)R_2(r)
=\frac{8\pi^2}{r}
\int_{r-1/2}^{1/2}\int_{r-s}^{1/2}\mathrm{d}s\mathrm{d}t\,s\omega(s)t\omega(t).
\end{equation*}
For example, consider a kernel of the form $\omega_\alpha(t)=\rho_\alpha e^{\alpha t}/(2\pi t)$, where $\alpha\in\mathbb{R}$ and $\rho_\alpha$ is a normalization constant, see \eqref{35-04}, and is equal to
\begin{equation*}
\rho_\alpha=\frac{\alpha^2}{(\alpha-2)e^{\alpha/2}+2}.
\end{equation*}
Note that the denominator turns to zero only at $\alpha=0$. Moreover, the asymptotic expansion near zero has the form $\alpha^2(1+o(1))/4$, and, therefore, the function $\rho_\alpha>0$ for all $\alpha\in\mathbb{R}$. 
Substituting the function $\omega_\alpha(\cdot)$ into the integrals, we obtain
\begin{align*}
A_3(x)R_{\alpha,1}(r)&=\frac{2\rho_\alpha^2}{\alpha^2r}\bigg[e^\alpha-e^{\alpha(1-r)}-\alpha re^{\alpha r}\bigg],\\
A_3(x)R_{\alpha,2}(r)&=\frac{2\rho_\alpha^2}{\alpha^2r}\bigg[e^\alpha-e^{\alpha r}-\alpha(1-r)e^{\alpha r}\bigg].
\end{align*}
Further, integrating and taking into account the continuity conditions
\begin{align*}
R_{\alpha,1}(r)\big|_{r=1/2-0}&=R_{\alpha,2}(r)\big|_{r=1/2+0},
\\
\partial_rR_{\alpha,1}(r)\big|_{r=1/2-0}&=\partial_rR_{\alpha,2}(r)\big|_{r=1/2+0},
\\
R_{\alpha,2}(r)\big|_{r=1-0}&=r^{-1}\big|_{r=1+0}=1,
\\
\partial_rR_{\alpha,2}(r)\big|_{r=1-0}&=\partial_rr^{-1}\big|_{r=1+0}=-1,
\end{align*}
we get an answer in the form
\begin{equation*}
G_{3,\omega_\alpha}(r)=\frac{1}{4\pi}
	\begin{cases}
		R_{\alpha,1}(r),	& r<1/2;  \\
		R_{\alpha,2}(r),
		& 1/2\leqslant r\leqslant1; \\
		r^{-1},	& r>1,
	\end{cases}
\end{equation*}
where auxiliary functions are defined as follows
\begin{align*}
R_{\alpha,1}(r)&=1+\frac{2\rho_\alpha^2}{\alpha^2}\bigg[\frac{e^\alpha}{2}(1-r)+
%\int_r^{1/2}\frac{\mathrm{d}u}{\alpha^2 u^2}\Big(e^{\alpha(1-u)}(1+\alpha u)-e^{\alpha}\Big)
\frac{e^\alpha}{\alpha^2}\bigg(2-2e^{-\alpha/2}-\frac{1-e^{-\alpha r}}{r}\bigg)
\\&
~~~~~~~~~~~~~~~~~~~~~~~~~~~~~-
%\int_r^{1/2}\frac{\mathrm{d}u}{\alpha^2 u^2}\Big(e^{\alpha u}\big((1-\alpha u)^2+1\big)-2\Big)
\frac{2}{\alpha^2}\bigg(2-2e^{\alpha/2}-\frac{1-e^{\alpha r}}{r}\bigg)
-\frac{1}{\alpha}\Big(e^{\alpha/2}-e^{\alpha r}\Big)
\\&
~~~~~~~~~~~~~~~~~~~~~~~~~~~~~+
\frac{(3+\alpha)}{\alpha^2}
\bigg(-1-e^{\alpha}+2e^{\alpha/2}\bigg)
+
\frac{1}{\alpha}\Big(e^{\alpha}-e^{\alpha/2}\Big)\\&
~~~~~~~~~~~~~~~~~~~~~~~~~~~~~-
\frac{1}{\alpha^2}\Big(e^\alpha-e^{\alpha/2}(2\alpha-4)-\alpha-5)\Big)
%\int_{1/2}^1\mathrm{d}u\,e^{\alpha u}
\bigg]
,\\
R_{\alpha,2}(r)&=1+\frac{2\rho_\alpha^2}{\alpha^2}\bigg[\frac{e^\alpha}{2}(1-r)+
\frac{(3+\alpha)}{\alpha^2}
\bigg(1-e^{\alpha}-\frac{1-e^{\alpha r}}{r}\bigg)
+
\frac{1}{\alpha}\Big(e^{\alpha}-e^{\alpha r}\Big)\\&
~~~~~~~~~~~~~~~~~~~~~~~~~~~~~+
\frac{1-r^{-1}}{\alpha^2}\Big(e^\alpha-e^{\alpha/2}(2\alpha-4)-\alpha-5)\Big)\bigg].
\end{align*}
The following equalities are of particular interest
\begin{align*}
||\omega||_{L_1}&=1,
\\
G_{3,\omega_\alpha}(0)&=\frac{\alpha}{2\pi}\frac{e^{\alpha}(\alpha-3)+4e^{\alpha/2}-1}{(e^{\alpha/2}(\alpha-2)+2)^2}\equiv\varphi_\alpha,
\\
A_3G_{3,\omega_\alpha}\big|_{r=0}&=\frac{\alpha^3}{2\pi}\frac{e^{\alpha}-1}{(e^{\alpha/2}(\alpha-2)+2)^2}=||\omega||_{L_2}\equiv\psi_\alpha.
\end{align*}
Let us take a closer look at properties of $\varphi_\alpha$ and $\psi_\alpha$. The function $\varphi_\alpha$ is strictly decreasing on $\mathbb{R}$. At the same time, we have
\begin{equation*}
\varphi_\alpha\big|_{\alpha\to-\infty}=+\infty
,\,\,\,
\varphi_\alpha\big|_{\alpha\to+\infty}=\frac{1}{2\pi}=G_3(1/2).
\end{equation*}
The function $\psi_\alpha>0$, for all values, has unique point $\alpha_c\approx3.72
>0$, in which the first derivative $\dot{\psi}_{\alpha_c}=0$ is equal to zero, strictly decreases on $(-\infty,\alpha_c)$, and strictly increases on $(\alpha_c,+\infty)$. At the same time, the limits have the form
\begin{equation*}
\psi_\alpha\big|_{\alpha\to\pm\infty}=+\infty.
\end{equation*}
This behavior is easily interpretable. Indeed, in the limit $\alpha\to-\infty$, the density of $\omega_\alpha(t)$ concentrates around $t=0$, which corresponds to the removal of regularization and, as a result, the transition $G_{3,\omega_\alpha}\to G_3$ to a function with a singularity at zero. In turn, the limit of $\alpha\to+\infty$ corresponds to a density concentration near $t=1/2$. Indeed, for every fixed $t\in[0,1/2)$ the kernel of the averaging operator decreases exponentially at $\alpha\to+\infty$ according to the rule
\begin{equation*}
\omega_\alpha(t)=\alpha e^{\alpha(t-1/2)}\big(1+o(1)\big),
\end{equation*}
while at the point $t=1/2$, the function $\omega_\alpha(1/2)$ increases proportionally to $\alpha$ in the main order, thereby leading to a $\delta$-shaped sequence, since the normalization remains fixed.
This situation symbolizes the transition to the $\delta$-functional, which leads to the minimum of $G_{3,\omega_\alpha}(0)$. However, the $\delta$-functional is not quadratically integrable, so this regularization is not applicable for all models equally. As an example of the application of the limiting case, we can give a sextic model, see \cite{Kh-2024,Kh-25}.

\subsection{The case $n=2$.}
\label{sec:35:3}
Substitute the index value, then $S_1=2\pi$, $S_0=2$, and formula \eqref{35-03} can be rewritten as
\begin{equation}\label{35-31}
g_2(r,s,t)=\frac{1}{2\pi^2}
\int_r^{t+s}\mathrm{d}u\,u^{-1}\int_0^{(u^2-(t-s)^2)/(4ts)}\mathrm{d}p\,p^{-1/2}(1-p)^{-1/2}.
\end{equation}
Next, using the value for the incomplete beta function $\mathrm{B}_r(1/2,1/2)=2\arcsin(\sqrt{r})$, see Paragraph 8.391 and formula 13 in Paragraph 9.121 of \cite{GR}, we get
\begin{equation}\label{35-32}
g_2(r,s,t)=\frac{1}{\pi^2}
\int_r^{t+s}\mathrm{d}u\,u^{-1}\arcsin\Big(\sqrt{(u^2-(t-s)^2)/(4ts)}\Big).
\end{equation}
Calculating the integral with such a density is a separate task, and, as in the case of $n=3$, only some special quantities are of particular interest. As a part of the study of a two-dimensional non-linear sigma model \cite{AIK-25}, the following values appeared
\begin{equation}\label{35-f-9}
\theta_j=\int_{\mathbb{R}^2}\mathrm{d}^2k\,|k|^{-2}\Big(\upsilon^{2j+2}(|k|)-\upsilon^{2j}(|k|)\Big),
\end{equation}
where $j\geqslant1$ and the auxiliary function is defined by the equality
\begin{equation*}
\upsilon(|k|)=\int_{\mathbb{R}^2}\mathrm{d}^2x\,e^{ikx}\omega(|x|).
\end{equation*}
Here, the density has properties from Theorem \ref{35-t-2} . Let us study a special case. To do this, we calculate the Fourier transform for the sharp cutoff function in the momentum representation, which, after the scale transformation from \eqref{35-06}, is the characteristic function of the ball $\mathrm{B}_1$,
\begin{equation}\label{35-e-1}
\frac{1}{2\pi}\int_{\mathrm{B}_1}\mathrm{d}^2k\,e^{ikx}=
\frac{1}{2\pi}\int_0^1\mathrm{d}t
\int_0^{2\pi}\mathrm{d}\phi\,te^{ist\cos(\phi)}=
\frac{J_1(s)}{s}
,
\end{equation}
where $|x|=s$, $|k|=t$, and $J_i(\cdot)$ is the Bessel function of the first kind. Thus, if the density of $\omega$ corresponded to the cutoff in the momentum representation, then it would be proportional to \eqref{35-e-1}. However, such a function has a non-compact support and, therefore, is not suitable as an averaging kernel for a cutoff in the coordinate representation. Let us cut it additionally manually. We define new kernel of the averaging operator by the equality
\begin{equation}\label{35-n-1}
\omega(s)=\frac{\alpha}{2\pi \kappa_\alpha}\frac{J_1(\alpha s)}{s}\chi(s<1/2),
\end{equation}
where $\alpha>0$, $\chi(\cdot)$ is the characteristic function of a set of $\mathbb{R}_+$, satisfying the equality, and also $\kappa_\alpha=1-J_0(\alpha/2)$ is a normalization factor, see \eqref{35-04}. It is clear that the kernel from \eqref{35-n-1} leads to a mixed regularization combining the properties of cutoffs both in the coordinate representation and in the momentum representation. Therefore, after substitution, we get
\begin{equation*}
\upsilon(\alpha t)=\kappa_\alpha^{-1}\int_0^{\alpha/2}\mathrm{d}s\,J_0(st)J_1(s)
\equiv\hat{\upsilon}(\alpha,t).
\end{equation*}
Function \eqref{35-f-9} with the new kernel will then be labeled with the index $\alpha$. Then, after scaling the integration variable, we get
\begin{equation*}
\theta_{j}^\alpha=2\pi\int_{\mathbb{R}_+}\frac{\mathrm{d}t}{t}\Big(\hat{\upsilon}^{2j+2}(\alpha,t)-\hat{\upsilon}^{2j}(\alpha,t)\Big).
\end{equation*}
We show that during in the limit $\alpha\to+\infty$, the values $\theta_j^\alpha$ tend to zero for all $j\in\mathbb{N}$. This fact can be proved in several steps. Let us define three auxiliary integrals for this
\begin{align*}
\theta_{j,1}^\alpha&=2\pi\int_1^{+\infty}\frac{\mathrm{d}t}{t}\Big(\hat{\upsilon}^{2j+2}(\alpha,t)-\hat{\upsilon}^{2j}(\alpha,t)\Big),\\
\theta_{j,2}^\alpha&=2\pi\int_{1/2}^{1}\frac{\mathrm{d}t}{t}\Big(\hat{\upsilon}^{2j+2}(\alpha,t)-\hat{\upsilon}^{2j}(\alpha,t)\Big),\\
\theta_{j,3}^\alpha&=2\pi\int_0^{1/2}\frac{\mathrm{d}t}{t}\Big(\hat{\upsilon}^{2j+2}(\alpha,t)-\hat{\upsilon}^{2j}(\alpha,t)\Big).
\end{align*}
It is clear that $\theta_{j}^\alpha=\theta_{j,1}^\alpha+\theta_{j,2}^\alpha+\theta_{j,3}^\alpha$. Moreover, each component, as well as their sum, can be estimated according to the following lemma.
\begin{lemma}\label{35-l-1}
For each fixed $j\in\mathbb{N}$, there are such numbers $\Theta_{j,i}>0$, where $i\in\{1,2,3\}$, and such a $N>1$ that for all $\alpha>N$ the inequalities are fulfilled
\begin{equation*}
|\theta_{j,1}^\alpha|\leqslant\frac{\Theta_{j,1}\ln(\alpha)}{\alpha},\,\,\,
|\theta_{j,2}^\alpha|\leqslant\frac{\Theta_{j,2}}{\alpha^{1/2}},\,\,\,
|\theta_{j,3}^\alpha|\leqslant\frac{\Theta_{j,3}\ln(\alpha)}{\alpha^{1/2}},\,\,\,
|\theta_{j}^\alpha|\leqslant\Big(\Theta_{j,1}+\Theta_{j,2}+\Theta_{j,3}\Big)\frac{\ln(\alpha)}{\alpha^{1/2}}.
\end{equation*}
\end{lemma}
\noindent The proof of this statement consists of five steps. Let us describe them in more detail.\\

\noindent\underline{Step 1.} Let us first show that there exist such $N>0$ and $C>0$ that for all $\alpha>N$ and $t\geqslant1/2$ the estimate is valid
\begin{equation}\label{35-f-4}
|\hat{\upsilon}(+\infty,t)-\kappa_{\alpha}\hat{\upsilon}(\alpha,t)|\leqslant 
\frac{1}{\pi\sqrt{t}}\bigg(\frac{\pi C}{\alpha}+
\big|\mathrm{si}(\alpha|1-t|/2)\big|+\big|\mathrm{ci}(\alpha(1+t)/2)\big|\bigg),
\end{equation}
where definitions were used for the integral cosine functions $\mathrm{si}(\cdot)$ and integral sine $\mathrm{ci}(\cdot)$, see Paragraph 8.230 in \cite{GR}. In particular, using the boundedness  of special functions on the interval under consideration, it can be argued that there exists such a $V_1>0$ that
\begin{equation}\label{35-f-5}
	|\hat{\upsilon}(+\infty,t)-\kappa_{\alpha}\hat{\upsilon}(\alpha,t)|\leqslant V_1.
\end{equation}
To derive the relations, we use formula 2 of Paragraph 6.512 from \cite{GR}
\begin{equation*}
\int_{\mathbb{R}_+}\mathrm{d}s\,J_0(st)J_1(s)=
\begin{cases}
	1,
	&t\in[0,1); \\
	1/2,	& t=1;\\
	0,
	&t>1. 
\end{cases}
=\hat{\upsilon}(+\infty,t),
\end{equation*}
where the value $\kappa_\alpha|_{\alpha\to+\infty}=1$ was additionally used in the last equation. Then we get the equality
\begin{equation}\label{35-f-1}
\hat{\upsilon}(+\infty,t)-\kappa_{\alpha}\hat{\upsilon}(\alpha,t)=
\int_{\alpha/2}^{+\infty}\mathrm{d}s\,J_0(st)J_1(s).
\end{equation}
Next, we apply formula 1 of Paragraph 8.451 from \cite{GR}, according to which there exist such $N_1>1$ and $C_1>0$ that for all $r>N_1$ and $l\in\{0,1\}$ an equation of the form holds
\begin{equation}\label{35-f-3}
J_l(r)-\sqrt{2/\pi r}\cos(r-\pi/4-l\pi/2)=\phi_l(r)r^{-3/2},
\end{equation}
where the absolute value of the continuous function $\phi_l(r)$ for all $r>N_1$ and $l\in\{0,1\}$ is bounded by the number $C_1$. For convenience, we will also assume that for all $r>N_1$, the inequality $|J_0(r)|\leqslant1/\sqrt{r}$ holds.
In this case, applying equality to $J_1(s)$, the difference from \eqref{35-f-1} for all $t\in\mathbb{R}_+$ and $\alpha>2N_1$ can be written out as
\begin{equation}\label{35-f-2}
\hat{\upsilon}(+\infty,t)-\kappa_{\alpha}\hat{\upsilon}(\alpha,t)=
\sqrt{\frac{2}{\pi}}
\int_{\alpha/2}^{+\infty}\frac{\mathrm{d}s}{\sqrt{s}}\,J_0(st)\cos(s-3\pi/4)+
\int_{\alpha/2}^{+\infty}\frac{\mathrm{d}s}{s^{3/2}}\,J_0(st)\phi_1(s).
\end{equation}
Let us strengthen the constraints on the parameter $\alpha$. Further, let $\alpha>4N_1$, then the ratio from \eqref{35-f-3} can be applied to $J_0(st)$. Then the right-hand side of the last equation is transformed to the sum of four terms
\begin{align*}
\hat{\upsilon}(+\infty,t)-\kappa_{\alpha}\hat{\upsilon}(\alpha,t)=&
\frac{2}{\pi\sqrt{t}}
\int_{\alpha/2}^{+\infty}\frac{\mathrm{d}s}{s}\,\cos(st-\pi/4)\cos(s-3\pi/4)+
\frac{1}{t^{3/2}}\int_{\alpha/2}^{+\infty}\frac{\mathrm{d}s}{s^3}\,\phi_0(st)\phi_1(s)\\
&\sqrt{\frac{2}{\pi t^3}}
\int_{\alpha/2}^{+\infty}\frac{\mathrm{d}s}{s^2}\,\phi_0(st)\cos(s-3\pi/4)+\sqrt{\frac{2}{\pi t}}
\int_{\alpha/2}^{+\infty}\frac{\mathrm{d}s}{s^2}\,\phi_1(s)\cos(st-\pi/4).
\end{align*}
Note that there is such a $C>0$ that the absolute value of the sum of the last three terms after evaluation modulo integral expressions does not exceed the value of $C/(\alpha\sqrt{t})$. Applying boundedness to the functions, we note that the following combination can be chosen as such a constant
\begin{equation}\label{35-f-10}
C=C_1^2/N_1^2+6C_1\sqrt{2/\pi}.
\end{equation}
Additionally transforming the product of cosines
\begin{equation*}
2\cos(st-\pi/4)\cos(s-3\pi/4)=\sin(s(1-t))-\cos(s(1+t))
\end{equation*}
and using definitions for the functions $\mathrm{si}(\cdot)$ and $\mathrm{ci}(\cdot)$, we get the stated estimate \eqref{35-f-4}, taking into account the found numbers $N=4N_1$ and $C$ from \eqref{35-f-10}.\\

\noindent\underline{Step 2.} Consider the integral for $\theta_{j,1}^\alpha$. Let the constraints on the $\alpha$ parameter stated in the previous step be true. Note that the absolute value of the integral expression, taking into account \eqref{35-f-5} and the equality $\hat{\upsilon}(+\infty,t)=0$ for $t>1$, is estimated as follows
\begin{align*}
2\pi\int_{1}^{+\infty}\frac{\mathrm{d}t}{t}\,\big|\hat{\upsilon}^{2j+2}(\alpha,t)-\hat{\upsilon}^{2j}(\alpha,t)\big|&\leqslant \kappa_\alpha^{-2j-2}V_1^{2j-1}(V_1+1)^2\times\\
&\times 2\int_{1}^{+\infty}\frac{\mathrm{d}t}{t^{3/2}}
\bigg(\frac{\pi C}{\alpha}+
\big|\mathrm{si}(\alpha|1-t|/2)\big|+\big|\mathrm{ci}(\alpha(1+t)/2)\big|\bigg).
\end{align*}
The first part of the integral can be calculated explicitly. The integral with the function $\mathrm{si}(\cdot)$ must first be divided into two parts by regions, for example, $[1,3]$ and $[3,+\infty]$, and then take advantage of the fact, see Paragraph 8.235 in \cite{GR}, that there are such $N_2>0$ and $C_2>0$, that the following estimates are correct for all $r>N_2$
\begin{equation*}
\big|\mathrm{ci}(r)\big|\leqslant C_2/r\,\,\,\mbox{and}\,\,\,\big|\mathrm{si}(r)\big|\leqslant C_2/r.
\end{equation*}
The third part of the integral is calculated without additional division. Next, we assume that in addition, $\alpha>N_2$ is executed. Then the following inequalities are true
\begin{equation*}
\int_{3}^{+\infty}\frac{\mathrm{d}t}{t^{3/2}}
\big|\mathrm{si}(\alpha|1-t|/2)\big|\leqslant\frac{2C_2}{\alpha}
\int_{2}^{+\infty}\frac{\mathrm{d}t}{t(t+1)^{3/2}}\leqslant\frac{C_2}{\alpha},
\end{equation*}
\begin{equation*}
\int_{1}^{3}\frac{\mathrm{d}t}{t^{3/2}}
\big|\mathrm{si}(\alpha|1-t|/2)\big|\leqslant\frac{2}{\alpha}
\int_{0}^{\alpha}\mathrm{d}t\,
\big|\mathrm{si}(t)\big|\leqslant\frac{2}{\alpha}\Big(N_2+C_2\ln(\alpha/N_2)\Big),
\end{equation*}
\begin{equation*}
\int_{1}^{+\infty}\frac{\mathrm{d}t}{t^{3/2}}
\big|\mathrm{ci}(\alpha|1+t|/2)\big|\leqslant\frac{2C_2}{\alpha}
\int_{2}^{+\infty}\frac{\mathrm{d}t}{t(t-1)^{3/2}}\leqslant\frac{C_2}{\alpha}.
\end{equation*}
Finally, collecting all the inequalities, we come to the conclusion that there exists such a $\Theta_{j,1}>0$, depending only on the index of $j$, that for $\alpha>\max\{4N_1,N_2\}$ an estimate of the form holds
\begin{equation*}
|\theta_{j,1}^\alpha|\leqslant\frac{\Theta_{j,1}\ln(\alpha)}{\alpha}.
\end{equation*}
\noindent\underline{Step 3.} Consider the integral for $\theta_{j,2}^\alpha$. In this case, it is necessary to take into account the fact that $\hat{\upsilon}(+\infty,t)=1$ for $t\in[0,1)$. Using the maximum value of the functions, we can verify the validity of the inequality
\begin{align*}
	2\pi\int_{1/2}^{1}\frac{\mathrm{d}t}{t}\,\big|\hat{\upsilon}^{2j+2}(\alpha,t)-&\hat{\upsilon}^{2j}(\alpha,t)\big|\leqslant \kappa_\alpha^{-2j-2}(V_1+1)^{2j}(V_1+1+\kappa_\alpha)\times\\
	&\times 2\int_{1/2}^{1}\frac{\mathrm{d}t}{t^{3/2}}
	\bigg(\frac{\pi C}{\alpha}+
	\big|\mathrm{si}(\alpha|1-t|/2)\big|+\big|\mathrm{ci}(\alpha(1+t)/2)\big|+\big|1-\kappa_\alpha\big|\bigg).
\end{align*}
Considering the intergal, it is clear that the first three parts can be studied in the same way as it was performed in the previous step. The fourth term is proportional to $|1-\kappa_\alpha|=|J_0(\alpha/2)|$ and therefore behaves like $\alpha^{-1/2}$, taking into account the expansion form \eqref{35-f-3}. Further, summing up all the relations, we come to the conclusion that there exists such a $\Theta_{j,2}>0$, depending only on the index of $j$, that for $\alpha>\max\{4N_1,4N_2\}$ an estimate of the form is performed
\begin{equation*}
	|\theta_{j,2}^\alpha|\leqslant\frac{\Theta_{j,2}}{\alpha^{1/2}}.
\end{equation*}
During the derivation, it was taken into account that $\alpha^{-1/2}$ tends to zero more slowly than $\ln(\alpha)/\alpha$, when $\alpha\to+\infty$.\\

\noindent\underline{Step 4.} Let us show the boundedness of the difference from \eqref{35-f-1} in the range $t\in[0,1/2]$. To do this, we use decomposition \eqref{35-f-2} and the auxiliary representation for the Bessel function, see formula 1 in Paragraph 8.411 of \cite{GR}, in the form
\begin{equation}\label{35-f-7}
J_m(st)=\frac{1}{\pi}\int_0^\pi\mathrm{d}\theta\,\cos(m\theta-st\sin(\theta)).
\end{equation}
Here, the index $m\in\mathbb{N}\cup\{0\}$. Decomposing the cosine product
\begin{equation*}
-2\cos(st\sin(\theta))\cos(s-3\pi/4)=\cos(s(1-t\sin(\theta))+\pi/4)+\cos(s(1+t\sin(\theta))+\pi/4)
\end{equation*}
and moving on to integration over $[-\pi,\pi]$, we come to the inequality
\begin{equation*}
\Big|\hat{\upsilon}(+\infty,t)-\kappa_{\alpha}\hat{\upsilon}(\alpha,t)\Big|\leqslant
\frac{2^{-1/2}}{\pi^{3/2}}\int_{-\pi}^\pi\mathrm{d}\theta\,\bigg|
\int_{\alpha/2}^{+\infty}\frac{\mathrm{d}s}{\sqrt{s}}\,\cos(s(1+t\sin(\theta))+\pi/4)\bigg|+\frac{4C_1}{\alpha^{1/2}}.
\end{equation*}
Notice that the replacement of the integration order was possible due to the fact that the integral expression always contains an oscillating exponent, since $t|\sin(\theta)|\leqslant1/2<1$ is executed for all values of $\theta\in[-\pi,\pi]$ and $0\leqslant t\leqslant1/2$. Indeed, we note that the integral 
\begin{equation}\label{35-r-1}
	\int_{\alpha/2}^{+\infty}\frac{\mathrm{d}s}{\sqrt{s}}\,\cos(s(1+t\sin(\theta))+\pi/4)=
	\frac{1}{\sqrt{\beta}}
	\Big(S\big(\sqrt{\alpha\beta}\big)-C\big(\sqrt{\alpha\beta}\big)\Big),
\end{equation}
where $\beta=(1+t\sin(\theta))/\pi$, after scaling of the variable can be calculated explicitly and expressed in terms of the Fresnel integrals $C(\cdot)$ and $S(\cdot)$, see Paragraph 8.250 in \cite{GR}. From this representation, it can be seen that the integral from \eqref{35-r-1} converges uniformly with respect to the parameter $\theta$ in the interval $[-\pi,\pi]$, and, therefore, the formula for replacing the order of integration is valid, see Theorem 4 of Paragraph 521 in \cite{34-f-111}. Taking into account additionally the fact that the Fresnel integrals are bounded functions on the real axis, we obtain the final statement. There exists a $V_2>0$ such that for all $t\in[0,1/2]$ and $\alpha>\max\{4N_1,4N_2\}$ the inequality holds
\begin{equation*}
\Big|\hat{\upsilon}(+\infty,t)-\kappa_{\alpha}\hat{\upsilon}(\alpha,t)\Big|\leqslant V_2.
\end{equation*}

\noindent\underline{Step 5.} Consider the integral for $\theta_{j,3}^\alpha$. Using the equality $\hat{\upsilon}(+\infty,t)=1$ for $t\in[0,1)$ and using the maximum value of the functions, we can verify the validity of the inequality
\begin{align*}
	2\pi\int_{0}^{1/2}\frac{\mathrm{d}t}{t}\,\big|\hat{\upsilon}^{2j+2}(\alpha,t)-\hat{\upsilon}^{2j}(\alpha,t)\big|\leqslant 2&\pi\kappa_\alpha^{-2j-2}(V_2+1)^{2j}(V_2+1+\kappa_\alpha)\times\\
	&\times \int_{0}^{1/2}\frac{\mathrm{d}t}{t}\,
	\bigg|
	\int_{\alpha/2}^{+\infty}\mathrm{d}s\,J_0(st)J_1(s)-J_0(\alpha/2)
	\bigg|.
\end{align*}
Let us integrate it in parts, then the density inside the module is transformed to the form
\begin{equation}\label{35-f-6}
J_0(\alpha/2)\big(J_0(\alpha t/2)-1\big)-t\int_{\alpha/2}^{+\infty}\mathrm{d}s\,J_1(st)J_0(s).
\end{equation}
It is clear that to complete the proof, it is sufficient to consider only two integrals. In the case of the first term, the main estimate follows from the following chain of inequalities
\begin{equation*}
\int_{0}^{1/2}\frac{\mathrm{d}t}{t}\,\Big|\big(J_0(\alpha t/2)-1\big)\Big|\leqslant
\int_{0}^{\alpha/4}\frac{\mathrm{d}t}{t}\,\Big|J_0(t)-1\Big|
\leqslant
\bigg(\int_{0}^{1}\frac{\mathrm{d}t}{t}\,\Big|J_0(t)-1\Big|+2\ln(\alpha/4)
\bigg)\leqslant2\ln(\alpha),
\end{equation*}
which are true for all $\alpha\geqslant4$.
We assume that $\alpha$ satisfies the conditions from the previous steps. Considering that $|J_0(\alpha/2)|$ is dominated by $2^{1/2}\alpha^{-1/2}$, we see that the first term in \eqref{35-f-6} leads to the dependence $\ln(\alpha)/\sqrt{\alpha}$. Next, consider the integral of the second term. Using \eqref{35-f-3}, we have
\begin{equation}\label{35-f-8}
\bigg|\int_{\alpha/2}^{+\infty}\mathrm{d}s\,J_1(st)J_0(s)\bigg|\leqslant\sqrt{\frac{2}{\pi}}\,
\bigg|\int_{\alpha/2}^{+\infty}\frac{\mathrm{d}s}{\sqrt{s}}\,J_1(st)\cos(s-\pi/4)\bigg|+
\frac{4C_1}{\sqrt{\alpha}}.
\end{equation}
Then using representation \eqref{35-f-7} for the Bessel function and performing transformations similar to those in step 4, we make sure that the integral on the right side of \eqref{35-f-8} has an estimate of the form
\begin{equation*}
\sqrt{\frac{2}{\pi}}\,\bigg|\int_{\alpha/2}^{+\infty}\frac{\mathrm{d}s}{\sqrt{s}}\,J_1(st)\cos(s-\pi/4)\bigg|\leqslant\max_{r\geqslant\sqrt{\alpha/(2\pi)}}\Big(
\big|S(r)-C(r)\big|+
\big|S(r)+C(r)-1\big|\Big).
\end{equation*}
With the help of properties of Fresnel integrals for large values of the argument, we we see that there is a behavior of the form $\alpha^{-1/2}$. Indeed, there are such $N_3>0$ and $C_3>0$ that for all $r>\sqrt{N_3/(2\pi)}$ the relations are fulfilled
\begin{equation*}
\big|C(r)-1/2\big|\leqslant C_3/r,\,\,\,
\big|S(r)-1/2\big|\leqslant C_3/r,
\end{equation*}
where the formulas of Paragraph 8.255 from \cite{GR} were used. Thus, collecting all the inequalities, we make sure that under the conditions imposed on $\alpha$, that is, under $\alpha>\max\{4N_1,4N_2,N_3\}$, there exists such a $\Theta_{j,3}>0$, depending only on the index of $j$, for which the inequality holds
\begin{equation*}
	|\theta_{j,3}^\alpha|\leqslant\frac{\Theta_{j,3}\ln(\alpha)}{\alpha^{1/2}}.
\end{equation*}

\section{Conclusion}
\label{35:sec:zak}

The averaging operator was considered in the context of problems related to the regularization of fundamental solutions. Several statements were proved, and three special cases that previously arose when studying the structure of singularities in quantum field models were examined. In particular, in Theorem \ref{35-t-1} , representations were obtained for fundamental solutions in a flat Euclidean space of arbitrary dimension, averaged twice over spheres of non-fixed radius, and the properties of smoothness and areas of monotonicity were studied. In turn, Theorem \ref{35-t-2} presented acceptable classes of kernels for averaging operators, depending on additional conditions.

Section \ref{35:sec:sp} dealt with special cases. The first of them concerns the situation when averaging over a ball turns into averaging over a sphere by concentrating density at the boundary. This leads to extreme values of the deformed fundamental solution at zero. The second case is related to a three-dimensional situation that occurs in the sextic model. In this case, a family of kernels and corresponding smoothed fundamental solutions were explicitly presented. The third option is devoted to the two-dimensional case that arises in the model of the principal chiral field. It was shown that the coefficients of renormalization (functionals of a special kind) in the case of using sharp cutoff in the momentum representation can also be obtained by limit transition from a cutoff in the coordinate representation.\\

\textbf{About other examples.} One of the interesting open problems is the study of special functionals for a mixed type of regularization (by analogy with Section \ref{sec:35:3}) in the four-dimensional case. Sets of such functionals are given in the appendix of \cite{Iv-2024-1}, devoted to the study of three-loop corrections.

\textbf{About extending the kernel class.} In this paper, we studied the case exclusively with probabilistic averaging, that is, when the kernel density is almost everywhere non-negative. Such a statement is more meaningful from a physical point of view. Nevertheless, mathematically, this limitation can be dispensed with. In this case, it would be interesting to study the properties (for example, monotonicity intervals) of the function $\mathbf{f}(\cdot)$, which may differ significantly due to negative values of the kernels.

\textbf{About the inverse problem.} In the course of the work, the properties of the function $\mathbf{f}(\cdot)$ were studied, which was explicitly built using the kernel $\omega(\cdot)$ of the averaging operator. In this regard, a fair question arises: "Under what conditions can the kernel be restored using the function $\mathbf{f}(\cdot)$, and what properties should be fulfilled?". The answer to this question is directly related to the positive definiteness of functions and Boas--Kac roots, see for example \cite{35-cc-1,35-cc-2}. This issue was not covered in this paper. Nevertheless, it is important and deserves attention. \\

\noindent\textbf{Acknowledgments.} A.V. Ivanov thanks the Ministry of Science and Higher Education of the Russian Federation, agreement number 075-15-2025-013, for their support, and expresses gratitude to N.V.Kharuk for useful comments on the sextic model, as well as to L.A. and K.A. for creating a positive working environment.

\vspace{2mm}
\noindent\textbf{Data availability statement.} Data exchange is not applicable to this article because no data sets have been generated or analyzed during the current study.

\vspace{2mm}
\noindent\textbf{Conflict of interest statement.} The authors claim that there is no conflict of interest.

\end{document}